\documentclass[]{article}
\usepackage{amsmath,amssymb,amsfonts}
\usepackage{amsthm}
\usepackage{algorithmic}
\usepackage{graphicx}
\usepackage{textcomp}
\usepackage{bmpsize}
\usepackage{xcolor}
\usepackage{lipsum}
\usepackage{svg}

\newcommand{\Aron}[1]{}
\newcommand{\Fuhito}[1]{}
\newcommand{\Go}[1]{}

\usepackage[bookmarks=false]{hyperref}

\usepackage[capitalise]{cleveref}

\newcommand{\change}[2]{{#2}}

\usepackage[backend=biber,style=ieee]{biblatex}
\def\BibTeX{{\rm B\kern-.05em{\sc i\kern-.025em b}\kern-.08em
    T\kern-.1667em\lower.7ex\hbox{E}\kern-.125emX}}

\let\xthefootnote\thefootnote
\newcommand{\memo}[2][red]{\def\thefootnote{\color{#1}\xthefootnote}\footnote{\color{#1}#2}}
\newcommand{\hiddenmemo}[2][black]{}
\newcommand{\void}[1]{}
\newcommand{\red}[1]{{{#1}}}

\newcommand{\CB}{\mathcal{B}}

\newcommand{\CBH}{{\hat{\mathcal{B}}}}
\newcommand{\CG}{\mathcal{G}}
\newcommand{\CM}{\mathcal{M}}
\newcommand{\CO}{\mathcal{O}}

\newcommand{\E}{\mathbb{E}}
\newcommand{\EH}{\hat{E}}

\newcommand{\LH}{\hat{L}}
\newcommand{\FH}{\hat{F}}
\newcommand{\HH}{\hat{H}}
\newcommand{\MH}{\hat{M}}
\newcommand{\I}{\sqrt{-1}}
\newcommand{\V}{\mathbb{V}}
\newcommand{\R}{\mathbb{R}}
\newcommand{\NRP}{Nakamoto Reward Process}

\newtheorem{prop}{Proposition}
\newtheorem{lemma}{Lemma}
\newtheorem{corollary}{Corollary}

\newtheorem{theorem}{Theorem}

\providecommand{\keywords}[1]
{
  \small	
  \textbf{\textit{Keywords---}} #1
}
    
\begin{document}
\void{
This extra page fixes the issue with the click-to-source feature due to margin notes. It should be removed right before submission.
\clearpage
\setcounter{page}{1}
}
\void{
\setlength{\marginparwidth}{1.545cm}
\pagestyle{plain}
}

\title{Equilibrium of Blockchain Miners \\with Dynamic Asset Allocation}

\author{{Go Yamamoto}$^{1}$, {Aron Laszka}$^{2}$, {Fuhito Kojima}$^{1}$\\
\small $^{1}${NTT Research Inc.} \\
\small $^{2}${University of Houston}\\
}
\date{}
%\IEEEoverridecommandlockouts
%\IEEEpubid{\makebox[\columnwidth]{978-1-7281-7091-6/20/\$31.00~\copyright2020 IEEE \hfill} \hspace{\columnsep}\makebox[\columnwidth]{ }}
\maketitle
%\IEEEpubidadjcol

\begin{center}
This paper is a revision and extension of our work that was published at 2nd Conference on Blockchain Research \& Applications for Innovative Networks and Services (BRAINS 2020)~\cite{yamamoto2020equilibrium}.
\end{center}
\vspace{1em}

\begin{abstract}
\Go{Rewrite abstract.  Downsize and make the focus in the conclusion}
We model and analyze blockchain miners who seek to maximize the compound return of their mining businesses.
%Since optimal strategies for repeated games require dynamic asset allocation, we extend existing mathematical models of the economics of blockchain mining to allow for dynamic mining algorithms.
The analysis of the optimal strategies finds a new equilibrium point among the miners and the mining pools, which predicts the market share of each miner or mining pool.
The cost of mining determines the share of each miner or mining pool at equilibrium.  We conclude that neither miners nor mining pools who seek to maximize their compound return will have a financial incentive to occupy more than 50\% of the hash rate if the cost of mining is at the same level for all. 
However, if there is an outstandingly cost-efficient miner, then the market share of this miner may exceed 50\% in the equilibrium, which can threaten the viability of the entire ecosystem.

\end{abstract}

\Go{Changing "balancing" to "allocation".  Google says this is more popular word.}

\section{Introduction}

\subsection{Background}
    \Fuhito{I think the term ``risk-averse'' is more standard.} 

\void{
\Aron{we can actually skip the entire first paragraph, this is all very widely known}
The Bitcoin paper by Satoshi Nakamoto introduced \change{an excellent}{a} combination of technology and \change{mathematics}{theory} to realize distributed consensus \change{among dynamically changing nodes}{in a public, open network}. 
Among the many contributions of Satoshi's paper, one of the essential contributions is its incentive design.
\Aron{this needs to be rewritten, the key idea is that no single entity (or group of colluding entities) controls more than 50\% of the mining power}The \change{basic}{key} idea to make the Bitcoin nodes secure against the 51\% attack is, with assuming a maximum $k$ dishonest mining power, gathering over $k$ honest mining power.  The protocol rewards a certain amount of Bitcoin to the winner of the Proof of Work (PoW) game to assure such honest mining power exists.  If the reward is \change{reasonable for miners}{high enough to make mining profitable}, the Bitcoin network can maintain enough \change{number of}{} honest miners. Thus, to make the Bitcoin network long-term sustainable, the reward should be \Aron{again, `reasonable' is not well defined in this context}reasonable for miners in the long-term. 
}
\change{In reality}{In Bitcoin network}, most of the mining power is controlled by mining pools, and most of the \Aron{mining power and hash rate are the same, but this sentence make it look like they are not. \textbf{Go: I think control of mining power is about the pool management, hash power generation means mining workers who possibly work for pools.}}hash rate is produced by ``mining factories'' that equip \change{(\Aron{what do you mean by fixed?}fixed)}{}ASIC mining machines~\cite{Taylor2017}. \Aron{not obvious how this follows \textbf{Go: How about this}} \change{Hence,}{Since ASIC mining machines have limited purposes other than mining,} the economic behavior of a miner depends on its long-term prediction rather than a short-term benefit. 
\Aron{how is this connected to the previous sentence?  \textbf{Go: I think the sentences is now connected.}}
\void{  \Fuhito{As \I wrote some time ago, I think the next sentence should be deleted.} We need a sustainable PoW-base cryptoasset.  
}
Hence, it is necessary to analyze long-term incentives that the PoW mechanism gives to each miner and to establish the theory to evaluate the sustainability of the cryptoasset. 
The theory is useful not only for miners to determine their strategy but also for protocol designers and business entities to use permissionless blockchains for their applications. There \change{is much}{exist prior work on the} \Aron{again, how is this related? \textbf{Go: Clarified.}}\change{game-theoretic}{}analysis of the relationship between the short-term change of \change{price}{mining cost} and mining power~\cite{MDOK2014, BankofCanada2018}\change{. However, it is not clear that the analysis is still valid for the long term. }{; however, the short-term analysis does not explain how people invest in mining assets for long-term profit.  We need a new approach for analyzing how the change of mining power affects profit on the long term.}

\subsection{Profitable Strategies for Repeated Games}
\Aron{needs to be connected to previous paragraph, we suddenly start talking about single-shot and repeated games without any explanation of why this is related to sustainability of mining}
To illustrate the difference between optimal strategies for \change{single-shot}{short-term  profit} and \change{}{those for long-term profit in} repeated games, we consider the following example.
Suppose that there is a repeated \change{coin-flip}{}game in which you \change{}{can bet a certain amount and flip a coin in every stage. You} gain 23\% of
\change{the}{your} bet with probability 1/2, while you lose 20\% of \change{the}{your} bet with \change{the rest of
the}{}probability \change{}{1/2}.  \change{You}{Suppose that at the beginning of the game, you} have 1,000 USD in your hand, and you can choose any amount of bet from your hand for each stage of the coin-flip~game. 

\change{}{First of all, if you keep betting a constant amount, then the strategy is likely to be outperformed in the long-term by a riskfree strategy that generates the compound return of the riskfree rate.  Hence betting in the game implies you will seek a better compound return rate than the riskfree rate.}
\change{If you bet all in,}{However, if you bet all in,}
 then you are likely to find your money is halved\Fuhito{``You will find your money is halved'' here is an overstatement because 45 times win and 45 times losses is not 100 percent probability event. Strictly speaking we will need to write down the distribution over final outcomes (or expectation or whatever criterion we opt to use), but if we want to be more informal --like what we seem to be doing here-- even in that case, we should sday something like ``You are likely to find your money is roughly halved''.} after about $90$ times because you are likely to result in $45$ wins and $45$ \change{loses}{losses}, so it will be \Aron{don't put linebreaks in the middle of a paragraph}\Aron{use $\cdot$ or $\times$ to denote multiplication} $1000 (1+ 23\%)^{45}(1-20\%)^{45} \approx 480$~USD.
\void{\change{Indeed, the}{The} time-aggregated return rate \change{is better to}{can} be calculated by log-return rate, and  \Aron{what is the base of this logarithm? if it is normal, then it should be $\ln$}$\log(1.23) \cdot 0.5 + \log(0.80) \cdot 0.5 \approx -0.0081$.
}

\Aron{this should not be in the introduction of a CS paper \textbf{Go: I agree it does not like a CS paper, but I think this paragraph is OK and is a good illustration because this paper is not for pure CS.}}
\change{
Is it possible to make a \change{profit}{positive compound return} from the game?  The answer is yes if you adjust the amount of bet.\Fuhito{Like I said before, this is not a fair comparison. Even in the above ``100 percent bet'' in the last paragraph, the amount of bet changes dynamically, although the fraction of the money bet is constant. \textbf{Go: Make sense.  I removed "dynamically" hare.  In general it is dynamic but not here.}}
}{
You can improve the compound return rate by adjusting the amount of bet.
}
Let \(W\) be the random variable that takes \(W=0.23\) with
probability \(1/2\) and \(W=-0.2\) with \(1/2\).
Let $f$ be a real number $0<f<1$ that decides the asset allocation for the bet:
If we have assets \(A\), you will bet \(fA\) in the game, and
will keep \((1-f)A\).
Your assets after the coin-flip game will be
\( (1-f) A + f (1+W) A,\)
so the probabilistic return rate is
\(X =  fW\).
\void{
The compound return rate \(X^{(k)}\) after \(k\) coin-flips is given
by
\(1+ X^{(k)} =  (1+X(1))(1+X(2))\cdots(1+X(k)).\)\Aron{is this period supposed to denote multiplication?}
 \(X(i)\) is an independent copy of \(X\) with possibly different \(f\)
for each \(i\).
}
To maximize the compound return rate, you would like to maximize
\(E[\log(1+X)] = \log( 1  + 0.23 f) 0.5 + \log(1 - 0.2 f ) 0.5.\)
By solving \(\frac{d}{df}{E[\log(1+X)]}=0\), it is easy to see the
maximum is attained when $f$ is \(f^* \approx 0.33\).  Your expected log-return rate at $f^*$ is about $0.0024$\change{.}{, a positive rate.}  \change{It is positive, so the game is profitable.}{}  

The choice of optimal allocation factor $f^*$ under the payoff function of the log-return rate is called the Kelly strategy~\cite{Kelly1956, Thorp1969}.  
The Kelly strategy is known to be the optimal asset allocation for repeated games with respect to the \change{}{expected} compound return rate~\cite{breiman1961, FinkelsteinWhitley1981}.

\subsection{Dynamic Asset Allocation in Blockchain Mining}
\change{Considering the repeated process of Proof-of-Work blockchain mining, we would like to ask a similar question.
Can we apply the Kelly strategy on the mining algorithms?  Indeed, this paper answers yes.}{In this paper, we apply the idea of Kelly strategy to Proof-of-Work blockchain mining.}
\change{However, to}{To} formulate the Kelly strategy in blockchain mining, we need \change{to establish}{}a model for the economics of blockchain mining \change{in a way}{} that allows dynamic mining algorithms.
To the best of our knowledge, the rewards of blockchain mining are modeled using the Poisson process (for example \cite{Rosenfeld2011}).
In the model, the miners are modeled as a fixed algorithm that receives probabilistic rewards according to the Poisson process.
This paper proposes a model that allows dynamic mining algorithms by formulating the economics of mining \change{in}{as} a binomial tree model for the probabilistic return from mining reward minus cost, while a Poisson process triggers the growth of the binomial tree.  We show that the proposed model is a generalization of the existing Poisson reward model.  

Using the \change{proposed}{new} model, we present an analysis of the decision on \change{the}{}dynamic asset allocation.
We call the mining strategy with the optimal \change{decision of the}{}asset allocation the \emph{growth-rate mining}.  
The analysis of the decision also finds the equilibrium point of hash rates among the growth-rate miners.

\subsection{Predicting Share of Bitcoin Mining Pools}
The equilibrium point predicts miners' shares of the hash rate of Bitcoin.
Assuming all the miners have the same cost rate of mining, the growth-rate miners will occupy about $9\%$ of the world hash rate when the world hash rate is at $85\%$ of the break-even point.  If a growth-rate miner is more cost-efficient by about $5\%$ than other growth-rate miners, then the miner will occupy about $13\%$.
 Interestingly, the
prediction \change{looks coincide}{coincides} roughly with reality.   

The growth-rate mining strategy does not threaten the security of Bitcoin \change{for}{with respect to} the $51\%$ attack as long as there are no miners who \change{realize an}{have} outstanding cost-efficiency.
However, the \Aron{what do you mean by ``bound of the mining cost for the $51\%$ attack''? \textbf{Go: It means the lower bound of the cost.  Changed to an explicit description.} }lower bound of the mining cost for the $51\%$ attack is not necessarily unrealistic in the current Bitcoin environment:
\change{The}{the} equilibrium point exceeds $50\%$ of the world hash rate if a miner is about $70\%$ more cost-efficient than other~miners.

\Go{Section Our contribution seems redundant.}

\subsection{Related Works}

\subsubsection{Known Equilibrium Points Among Miners}
 Chiu et al.\ claim a Nash equilibrium point from the Cournot game setting the mining reward minus cost as the payoff function for Player~$i$~\cite{BankofCanada2018}.  In our notations defined in Section \ref{sec:model}, the hash rate at equilibrium is given by
\Aron{these symbols have not been introduced and thus are meaningless to the reader \textbf{Go: Updated}}$ \MH_i = \frac{m-1}{m^2} \frac{B}{c}$ when $c_i = c$ for all $i$.

Pagnotta et al.\ claim equilibrium points of hash rates based on miners' profits~\cite{Pagnotta2018}.  They focus on the analysis of Bitcoin's equilibrium price under a model of the miners' network.  The hash rate at the equilibrium is essentially similar to that described in~\cite{BankofCanada2018}.

Cong et al.\ claim equilibrium points of mining pools' size based on the CARA utility function~\cite{Cong2018}, and the equilibrium point depends on an exogenous parameter indicating the degree of risk-aversion.

\subsubsection{Mining Pools}

% Pool Strategies Selection in PoW-Based Blockchain Networks: Game-Theoretic Analysis
Wang et al.\ consider the mining pools' choice between being open or closed to miners: the former strategy is likely to be more efficient since attracts more miners, while the latter strategy protects the pool from certain attacks~\cite{wang2019pool}.
The authors model the pools' choice as a two-stage game, in which pools choose to be open or not and to attack or not, and find that weaker pools are more likely to attack.

% Research on the Selection Strategies of Blockchain Mining Pools
Qin et al.\ study how miners select which mining pool to join, considering pay-per-share, pay-per-last-$N$-share, and proportional reward mechanisms~\cite{qin2018research}.
The authors model pool selection as a risk decision problem based on maximum-likelihood criterion, which can provide managerial insights for miners.
% Evolutionary Game for Mining Pool Selection in Blockchain Networks
Liu et al.\ study the dynamics of mining pool selection and find that the hash rate for puzzle-solving and block propagation delay are the two major factors that determine the results of the competition between mining pools~\cite{liu2018evolutionary}.

% Incentive Compatibility of Bitcoin Mining Pool Reward Functions
Schrijvers et al.\ study the incentive compatibility of mining pool reward mechanism using a game-theoretic model, in which miners can choose between reporting or delaying when they discover a share or full solution~\cite{schrijvers2016incentive}.
The authors show that proportional rewards are not incentive compatible, but the pay-per-last-$N$-shares mechanism is in a more general model, and they introduce a novel incentive compatible mechanism.

%\section{Probabilistic process for one-shot mining}
%\input{1_oneshotmining}

\section{Mining with Dynamic Asset Allocation}
\label{sec:model}
\subsection{Model}%\label{sec:model}

\Aron{there should text with sentences here}
We model the economics of blockchain mining as a repeated reward process.

\textbf{Environment}
\begin{itemize}
    \item  \(B\) is the reward for mining the next block.
\item \(\tau\) is the average time interval between new block arrivals.
\item \(r\) is the riskfree rate for time interval $\tau$, usually the interest rate for Treasury bonds.
\end{itemize}

\hypertarget{header-n2}{%
\textbf{Players}\label{header-n2}} \change{}{There are a finite number of Players.}
\change{Player}{Each Player}~\(i\) has the following parameters.
\begin{itemize}
\item
  A set of balance sheets $\CO_i \subset \CO$ \change{that Player chose}{which the Player chooses} one from.   $\CO$ is the set of all possible balance sheets $\CO =\{ (E, L, M, F)\in \R_{\ge0}^4 \}$ \change{}{that} \change{satisfies}{\Aron{since subject is balance sheet, not the set}satisfy} 
   \(E + L = M + F \), and \(L=0\) if \(F \ne 0\).
 We call \(E\) Equity, \(L\) Liabilities, \(M\) Mining Assets, and \(F\) Riskfree Assets. 
\item
  Price of facility \(d_i\), \change{}{which is} the average price for facilities that produce the unit hash rate.
\item
  Cost rate \(c_i\), \change{}{which is} the average cost for running the device per the unit hash rate per time interval of $\tau$.
\end{itemize}

\textbf{Shared information}
The parameters for the Environment are publicly known.
Also, each Player's existence and its hash rate of $M_i/d_i$ are publicly known.

\Go{The Nakamoto game will be simplified.  Keep the current Single Nakamoto Game with interest payments as the Nakamoto Game, and we can say the repeated Nakamoto Game for the full mining }
\hypertarget{header-n33}{%
\textbf{\change{Single Nakamoto Game}{\NRP}}\label{header-n33}}
The \change{Single Nakamoto Game}{\NRP} is \Fuhito{The word ``game'' is saved for a specific meaning in game theory, and in particular we don't use that word for a random event. Changing ``is'' to ``involves'' is bad because strictly speaking it won't give a definition of what Nakamoto Game is, but I think this is better than ``is'' nevertheless.  \textbf{Go: Agreed that involves is not good.  How about the new name?}} a timeless random event for Players with balance sheets.  \red{Of all Players who play the game, Player $i$ with Mining Assets $M_i$ is exclusively randomly chosen with probability
\[p_i=\frac{M_i/d_i}{\sum_j{M_j/d_j}},\]
and obtains revenue of \(B\). 
In addition, each Player $j$ always \change{pay}{pays} cost \( c_j M_j/d_j\).}  We call \(p_i\) the \textit{success probability} of mining for Player \(i\).
Let $R_i$ be the random variable for the revenue minus the cost.
We call $R_i$ the return of the \change{Single Nakamoto Game}{\NRP}.

\hypertarget{header-n71}{%
\textbf{Nakamoto Game}\label{header-n71}}
The Nakamoto Game is a repeated game with the following stage game in finite time interval from $t=0$ to $t=T$.
Let $\Theta$ be the Poisson process with \(\lambda = \frac{\sum_j{M_j/d_j}}{D}\).
\(D\) is a parameter adjusted so that $D$ is close to $ \tau \sum_j{M_j}/d_j$.
\hiddenmemo{Fuhito: I'm not exactly sure, but my guess is that $D$ is assumed to be exactly equal to $ \tau \sum_j{M_j}/d_j$? By the way, $d_i$ is assumed to be the same among players $i$ and then normalized at $1$ soon, so I think it's better to just erase $d$ throughout (if I'm missing some other place in which different $d_i$ becomes relevant, sorry and I understand we want to include it). Indeed, later we refer to $M_i$ directly as a hash rate, so I think it's better to treat $M_i$ itself as the hash rate from the beginning. \\
\textbf{ Go: I prefer keeping $d_i$ here because readers will not be familiar with notations that identify the Mining Assets with the hash rate.  Later we drop $d$ by declaring we assume all the Players are homogeneous.}
}

\begin{enumerate}
\def\labelenumi{\arabic{enumi}.}
\item
  For all $i$, Player $i$ chooses balance sheet $\CB_i \in \CO_i$.
\item
  Wait a trigger according to $\Theta$.
\item  \red{All the Players execute the \change{Single Nakamoto Game}{\NRP}.  For all $i$, Player $i$ with the Mining Assets $M_i$ of $\CB_i = (E_i, L_i, M_i, F_i)$ obtains the return $R_i$.}
\item For all $i$, Player $i$ pays (or receives) interests $r (L_i - F_i)$.
%\item
%  Go to Step 1.
\end{enumerate}
\hiddenmemo{Fuhito: I wonder if we can switch the order of items 1 and 2 in the above description. This is because it seems closer to what we think of as a ``realistic'' model of games played by the bitcoin miners. In the present order, a randomly realized trigger (which is like ``successful mining'' lets all players choose how much hash power to expand at that moment, but then this amount affects who wins the present bitcoin mining. Surely this model is an abstraction anyway, but to the extent possible, I think it's better to match the timing closer to what's going on in real bitcoin mining. (If I still misunderstood what you are doing here, my apologies)\\
\textbf{GO: I modified so that Players choose BS before waiting for the trigger.  I considered our convenience in reusing this model in our next paper.  For example, I can imagine the analysis of the model in which the trigger is propagated with delay or intercepted by some malicious party.  It will explain the economics of the selfish mining attacks and related ideas.}} 
The payoff for each stage game of the Nakamoto Game is  $\log(1+ \frac{R_i -r(L_i-F_i)}{E_i})$, and the payoff for the Nakamoto Game is the sum of the payoff of the stage games.

When the Players have the same facility price \(d = d_i\), we say that the
 Players are \textit{homogeneous}\change{, and we normalize the units}{. In this case, we normalize the prices without loss of generality} so that $d=1$.
We say Player is \textit{static} if it always chooses the fixed balance sheet.

\subsection{Assumption on the Variance of Cost Rate}
The Nakamoto Game models the mining cost rate as a constant for each new block arrival,
ignoring the timewise variance of the cost.  \change{It is because of the following reasons:}{This assumption is realistic for two reasons.}
\change{(1) Most}{First, most} of the \Aron{what does this have to do with the cost rate? \textbf{ Go: added mention. actual ratio of the two sources of variance is about 1:share. }}variance of the return in the one-shot mining comes from the variance from the mining reward\change{, and
(2)}{ and that from the cost rate is minor in practical settings.  Second,} we are interested in the behavior of Players that \change{remains}{remain} robust to change of external factors in block arrival timing such as other miner's behavior with possible malicious intentions, the delay of block propagation network, possible forks, and so on.

\subsection{Return from the \change{Single Nakamoto Game}{\NRP}}
\hiddenmemo{Fuhito: I am still confused about what the content of this subsection is, especially Theorem 1. The only difference between the ``Nakamoto Game'' and the ``Poisson reward model'' is, as far as I can tell, that in the former, each Player is allowed to choose the balance sheet, while in the latter, the hash rate is exogenously given. Isn't it obvious from this fact that the rewards in these two processes are equivalent, or am I missing something more substantial here?\\
\textbf{Go: The two models are different formulations, and the relevance is non-trivial.   The revenue of the Nakamoto Game is "Wait for $1/\lambda$ minutes on average, and draw a random binomial process to gain the reward or none."  The Poisson model is "Wait for $1/\lambda_{i}$ minutes on average, and always obtain the reward".  The former is the compound Poisson process, and the latter is the standard Poisson process.  Fuhito's comment will be reasonable if we look at the 1st moment only.  The correspondence at the 2nd moment and the higher are not trivial.}

\textbf{  For example, we can model the cost in the Poisson reward model by $-cT$, $T$ is the time the game ends.  We can import this cost term to the Nakamoto Game using the equation on the moment generating functions as discussed in the paper, but the corresponding cost term in the stage game is the derivation of the Dirac delta function.}

}
For a random variable $X$, the moment generating function of $X$ is defined as \(\CM_X(u) = \E[e^{uX}]\).
\void{When $X$ is valued on non-negative integers, the probability generating function of $X$ is defined as
\(\CG_X(s) = \CM_X(\log(s)).\)
}

\begin{prop}
Given balance sheets $\CB_i =(E_i, L_i, M_i, F_i) \in \CO_i$ for all $i$, let $R_i$ be the random variable for the return from the \change{Single Nakamoto Game}{\NRP} for Player~$i$.  Then,
\(\CM_{R_i}(u) = \CM_{\text{Revenue}}(u) \CM_{\text{Cost}}(u)\)
for 
\(\CM_{\text{Revenue}}(u) = p_i e^{u  B} + (1-p_i) \)
and
\(\CM_{\text{Cost}}(u) = e^{-u c_i M_i/d_i}.\)
\end{prop}
\void{
We obtain the proposition because  $\CM_{X+Y}(u) = \CM_X(u)  \CM_Y(u)$ for independent random variables $X, Y$.
}

\begin{corollary}\label{cor:EV}
 For the return $R_i$ of the \change{Single Nakamoto Game}{\NRP} with homogeneous Players,
\[\E[R_i ] = \frac{ B M_i}{\sum_j{M_j}} - c_i M_i,
\text{and}\ 
\V[R_i] = \frac{ B^2 M_i M_{-i}}{(\sum_j{M_j})^2}\]
for \(M_{-i}= \sum_{j\ne i}{M_j}\).
\end{corollary}
\Go{Add some descriptions to help readers on math ops}

\begin{corollary}
Let $Y_i = B/c_i$.  Player $i$'s response in \change{Single Nakamoto Game}{Nakamoto Game} satisfies $M_i \le Y_i - \sum_{j\ne i}{M_j}$ if $(0,0,0,0)\in \CO_i$.  In particular $M_i=0$ if $Y_i \le \sum_{j\ne i}{M_j}$.
\end{corollary}

We call $Y_i$ the break-even hash rate for Player $i$.  
\hiddenmemo{
If we include time-variance in the Single Nakamoto Game, then the moment function for cost is
\[\CM_{\text{Cost}}(u) = -\frac{u c_i M_i}{\lambda\CM_{\text{Revenue}}(u)}+1.\]
We omit $d_i$ in this note.  We should set $B=0$.  Actually the cost term of the Single Nakamoto Game that corresponds to the naive cost at the Poisson reward model is dependent to $B$.  

This is deduced because $\CG_{N(\Theta)}(\CM_{\text{cost}}(u)) = e^{-uTc_i M_i}$ when running the process from $t=0$ to $t=T$. That is,
\[T\lambda (\CM_{\text{cost}}(u) -1) = {-uTc_i M_i}.\]
The corresponding probability density function is the derivation of the delta function.
\[\rho(x) = (-uc_iM_i/\lambda)\delta'(x)+ \delta(x).\]

The variance will have an additional term.
\[\E[R_i ] = \frac{ B M_i}{\sum_j{M_j}} - c_i M_i\]
\[\V[R_i] = \frac{ B^2 M_i M_{-i}}{(\sum_k{M_k})^2} + c_i^2 M_i^2\]

If we keep this $c_i^2$, then the calculus on $\CB^*$ will not allow practical expression.
We can ignore $c_i^2$ because it is very small, and it makes the calculus simple.
}

\subsection{Nakamoto Game's Reward Model and Poisson Reward Model}
%\section{Moment generating function for Nakamoto reward model}
%\section{Nakamoto Reward model and the rewards by the Poisson distribution}

\void{
\begin{prop}
Given balance sheets $\CB_i =(E_i, L_i, M_i, F_i) \in \CO_i$ for all $i$, let $R_i$ be the random variable for the return from the Single Nakamoto Game for Player $i$.  Then,
\[\CM_{R_i}(u) = \CM_{\text{Revenue}}(u) \CM_{\text{Cost}}(u)\]
for 
\[\CM_{\text{Revenue}}(u) = p_i e^{u  B} + (1-p_i) \]
and
\[\CM_{\text{Cost}}(u) = e^{-u c_i M_i/d_i}.\]
\end{prop}
We obtain the proposition because  $\CM_{X+Y}(u) = \CM_X(u)  \CM_Y(u)$ for independent random variables $X, Y$.

\begin{corollary}\label{cor:EV}
 For the return $R_i$ of the Single Nakamoto Game with homogeneous Players,
\[\E[R_i ] = \frac{ B M_i}{\sum_j{M_j}} - c_i M_i\]
and
\[\V[R_i] = \frac{ B^2 M_i M_{-i}}{(\sum_j{M_j})^2}\]
for \(M_{-i}= \sum_{j\ne i}{M_j}\).
\end{corollary}

}

Since the original Bitcoin paper~\cite{nakamoto2009bitcoin}, the reward for miners has been modeled using the Poisson process\cite{Rosenfeld2011}.  We call it the Poisson reward model.

\textbf{Poisson Reward Model}
A player with hash rate $M_i/d_i$ will receive the revenue $B$ according to the Poisson process with $\lambda_i = \frac{M_i}{D d_i}$. 
$D$ is the difficulty parameter, which is adjusted so that $\lambda^{-1} = \tau$ for $\lambda=\frac{H}{D}$, where $H$ is the world hash rate.

We claim that the revenue of the Nakamoto Game with static Players replicates that of the Poisson reward model. 
For a random variable $X$, the moment generating function of $X$ is defined as \(\CM_X(u) = \E[e^{uX}]\).
When $X$ is valued on non-negative integers, the probability generating function of $X$ is defined as
\(\CG_X(s) = \CM_X(\log(s)).\)

\begin{prop}
Suppose static Players play the Nakamoto Game and Players $i$ chooses balance sheet $(E_i, L_i , M_i, F_i)$ for time interval \(t=0\) to \(t=T\).
Let $\CM^i(u)$ be the moment generating function for the sum of the return of player $i$.  Then,
\[\CM^i(u)= e^{T\lambda\left( (p_i e^{u B} + 1-p_i) e^{-u c_iM_i/d_i}-1\right)}\]
\end{prop}
\begin{proof}
Let $N(\Theta)$ be the random variable given by the number of triggers according to $\Theta$.
Since $\Theta$ is the Poisson process, we have
\[\CG_{N(\Theta)}(s) = e^{T\lambda(s-1)}.\]
Every Single Nakamoto Game is independent identical, so the sum of the return follows the compound Poisson distribution, 
\begin{align*}
\CM^i(u) &= \CG_{N(\Theta)}(\CM_{R_i}(u)) \\
&= e^{T\lambda\left( (p_i e^{u B} + 1-p_i) e^{-u c_iM_i/d_i}-1\right)}.
\end{align*}
\end{proof}

We obtain the moment generating function of the Poisson reward model if the cost rate is set $c_i=0$ in the Nakamoto Game.  It implies there is an interpretation from the analysis over the Poisson reward model to those over the Nakamoto Game.
\begin{theorem}\label{thm:poisson}
The random variable for the revenue in the Nakamoto Game by homogeneous static Players exactly coincides with the random variable for the revenue in the Poisson reward model.
\end{theorem}

\begin{proof}
Let $\CM^i_\text{Revenue}(t)=  \CM^i(t)|_{c_i =0}$.  Then 
\begin{align*}
    \CM^i_\text{Revenue}(u)%& =  e^{T\lambda p_i (s^{\gamma_i B }-1)}\\
    &= e^{T \lambda_i( e^{u  B} -1)}
\end{align*}
for $\lambda_i = \frac{M_i}{Dd_i}$.  This coincides with the moment generating function for the revenue at the Poisson reward model.  It implies the random variables for the revenue exactly match because
the moment generating function determines the corresponding probabilistic distribution if exists.
For a moment generating function $\CM$,
$$
\rho_X(x) = \frac{1}{2\pi}\int_\R{e^{-\I u x} \CM(\I u) du}
$$
produces the random variable $X$ whose probability density function is $\rho_X$, and $X$ satisfies $\CM = \CM_X$.
\end{proof}

\section{Decision of Asset Allocation and Finance}
\label{sec:decision}
\void{
This section provides the following theorems.
First, we apply the same method as in \cite{Thorp2006} on the binomial tree of the Nakamoto Game, and obtain the optimal asset allocation decision.  Then, we add an analysis of the optimal financial decision of the miners from the specific constraints of the blockchain mining.   Combining, we obtain observations that predict the reasonable behavior of miners.
}
We assume the Players are homogeneous hereafter.

\begin{theorem}\label{thm:decisionB}
Suppose Player $i$ plays the Nakamoto Game with $\CO_i =\CO$.   Fix Player $j$'s balance sheet $\CB_j$ for $j\ne i$.
Then, there is a unique balance sheet $\CB^*_i= (E^*_i, L^*_i, M^*_i, F^*_i)\in \CO_i$ that maximizes the expected payoff from the stage game of the Nakamoto Game.   Player $i$ maximizes  the expected payoff of the Nakamoto Game by continuously choosing $\CB^*$ for each stage game.  $\CB^*$ is determined by
\(M^*_i = \frac{1}{3}\left(\frac{ B}{c_i+r} - M_{-i}\right).\)
\end{theorem}

\begin{theorem}\label{thm:equilibrium}
Suppose $(m+n)$ Players play the Nakamoto Game, and let $I= \{1,2,\cdots,m\}$ and $K = \{m+1, m+2,\cdots, m+n\}$.    Suppose that for $i\in I$,  Players $i$  tries to maximize the expected payoff of the Nakamoto Game with $\CO_i = \CO$.  For $k\in K$, Player $k$ has only choice of the balance sheet with $\CO_k =\{ \CB_k\}$.
Then, there is an equilibrium point $\CBH_i=(\EH_i, \LH_i, \MH_i, \FH_i)$ among Players $i \in I$ in which $\CBH_i$ is determined by
\[\MH_i = \left(  \frac{1}{c_i + r} - \frac{1}{m+2}\sum_{j\in I}{ \frac{1}{c_j + r}}\right)\frac{B}{2}- \frac{Z}{m+2}\]
for $Z = \sum_{k\in K}{M_k}$. 
\hiddenmemo{Fuhito: Is $Z$ here the same thing as $M_0$ in Section 3.2? If so, how about using $M_0$ throughout? Go: Yes. $M_0$ is introduced for the convenience of notation, but it is not good for the theorem statement.  If you know better notation please fix it. \textbf{Fuhito}: I actually thought $M_0$ is fine. Another possibility is to write $\bar M$ --- the bar notation often times indicates that it's an exogenously fixed constant (at least in economics writing style)
\textbf{Go:I replaced $M_0$ with $Z$, and the exogenous hash rate is now formulated as those from exogenous players.  I think it is more realistic, and technically $M_0$ is inconvenient to describe the index of players when we take the sum. The poor exogenous players do not have freedom to choose the zero BS $(0,0,0,0)$ for now.  We could improve the analysis in the future work so that we can discuss the entrance/leave of players.}
}
\end{theorem}

\subsection{Optimal Asset Allocation}\label{OptimalF}
\void{
Let \(X\) be a random variable for return from biased coin-flipping processes.
The Taylor series for \(\log(1+x)\) around \(x=0\) gives an
approximation,
\begin{align*}\E[\log(1+X)] &= \E[ X - X^2/2 + \cdots]\\&=  \E[X] - \frac{\V[X]}{2} + O(\E[X]^2).
\end{align*}
}

\Go{State explicitly there is noting new about quantitative finance, and downsize this section.  The purpose of this section is preparing the standard notation for readers because Thorp's context is not necessarily common knowledge.}

\begin{prop}
Let \(W_i = R_i/M_i\), the random variable for the return rate
over Mining Assets from the \change{Single Nakamoto Game}{\NRP} with Mining Asset $M_i$ for Player $i$.
For given $\CB_i =(E_i, L_i, M_i, F_i)$,
the payoff of the stage game of the Nakamoto Game is given
by $\log(1+X_i)$ for
$X_i = (1-f)r + f W_i,$
\change{}{where} \(f\) is the
leverage rate \change{$M_i = f E_i$}{$f = M_i / E_i$}.
\end{prop}
\void{
\begin{IEEEproof}
The player will receive (or pay if negative) the risk-free interest $r(F_i - L_i)$.
Recall $E_i + L_i = M_i + F_i$.
If $f>1$, then $F_i=0$ and $L_i = M_i - E_i =(f-1) E_i$. Player will receive $-r L_i = (1-f)r E_i$.
If $0<f<1$, then $L_i=0$ and $F_i = E_i- M_i = (1-f)M_i$.  Player will receive $r F_i = (1-f) r E_i$
\end{IEEEproof}
}

The expected payoff of the stage game is approximated as
\begin{align*}\E[\log(1+X_i)] &= r - \frac{r^2}{2} + f(\E[W_i] - r) - \frac{f^2 \V[W_i]}{2} + O(\E[W_i]^3).\end{align*}
The proof is in Appendix \ref{appendix:approx}.
We use more accurate approximation formula than described in \cite{Thorp2006} because the optimal $f$ is $O(\E[W_i])$ in our situation.
\void{We can apply the same formula \change{works for}{can be applied to} the binomial tree models to estimate the payoff of the stage games of the Nakamoto Games. }
\hiddenmemo{Fuhito: If we take seriously that, in Nakamoto Game, players engage in repeated games, it's unclear that the equilibrium identified in this paper is the unique equilibrium. (Recall our earlier discussion on ``Folk theorem''). 
\textbf{Go: We would like to make it explicit that the Nakamoto Game is played with an arbitrary, but finite time interval.  I agree that it is an independent interesting topic to analyze the equilibrium of miners in mining process that are assumed to continues forever or with indefinite termination.}}

Let \(g_\infty(f)= r - \frac{r^2}{2}+ f(\mu - r) - \frac{f^2 \sigma^2}{2}\) for \(\mu = \E[W_i]\)
and \(\sigma^2 = \V[W_i]\).
The player would like to maximize \(g_\infty(f)\). Solving
\(g_\infty'(f) = 0\), \(g_\infty(f)\) attains the maximum at \(f= f^*\)
for
\(f^* = \frac{\mu-r}{\sigma^2},\)\label{eqn:optf}
and
\(g_\infty(f^*) = \frac{S^2}{2} + r - \frac{r^2}{2}\)
for \(S = \frac{\mu-r}{\sigma}\). \(S\) is called the Sharpe ratio of \(W_i\).
Thus we obtained the following proposition.

\begin{prop}\label{prop:uniquebalance}
Suppose Player $i$ seeks the optimal balance sheet and fix the balance sheet $\CB_j$ for every $j \neq i$.
(1) Given Mining Assets $\tilde{M}$ there is a unique balance sheet $\CB_i \in \CO(\tilde{M})$ that maximizes the expected payoff of the stage game of the Nakamoto Game with the balance sheet chosen from $\CO(\tilde{M})= \{(E,L,M,F)\in \CO| M = \tilde{M}\}$. 
(2)
The maximal expected payoff of the stage game of the Nakamoto Game for Player~$i$ with balance sheet $\CB_i$ is approximated by $ {S_i^2/2 + r - \frac{r^2}{2}}$.  $S_i$ is the Sharpe ratio of the return rate of the \change{Single Nakamoto Game}{\NRP},  namely
$S_i = \frac{\E[W_i]-r}{\sqrt{\V[W_i]}}.$
\end{prop}
\void{
\begin{IEEEproof}
Set $M_i = \tilde{M}$ and find $f^* = (\E[W_i]-r)/\V[W_i]$.  If $0<f^*<1$, then $L_i=0$ and $E_i= M_i/f^*$, so  $F_i = E_i- M_i$.
If $1<f^*$, then $F_i=0$ and $E_i=  M_i/ f^*$, so $L_i = M_i-E_i $.
If $f^*<0$, the mining will not be profitable, so the balance sheet is the same as the one when $f^*=0$.  
\end{IEEEproof}
}

%\subsection{The equilibrium among miners}

\subsection{Equilibrium among Players}
Suppose $(m+n)$ Players play the Nakamoto Game.  Let $I= \{1,2,\cdots,m\}$ and $K = \{m+1, m+2,\cdots, m+n\}$.    Suppose that for $i\in I$,  \change{Players}{Player} $i$  tries to maximize the expected payoff of the Nakamoto Game with $\CO_i = \CO$.  For $k\in K$, Player $k$ has only choice of the balance sheet with $\CO_k =\{ \CB_k\}$.
Let $Z =\sum_{k\in K}{M_k}$, $H = \sum_{j\in I}{M_{j}} + Z$ and $M_{-i} = -M_i + H$ for $i\in I$.
The Sharpe ratio of the return rate of Player \(i\) in the \change{Single
Nakamoto Game}{\NRP} is
\[S_i = \frac{1 - \frac{c_i + r}{ B}(M_i  + M_{-i})}{\sqrt{M_{-i}/M_{i}}}\]
as Corollary \ref{cor:EV} implies
\(\E[W_i] =  \frac{ B}{H}-c_i,
\V[W_i] = \frac{ B^2 M_{-i}}{M_i H^2} .\)

By solving \(\frac{\partial S_i}{\partial M_i}=0\), we obtain \(M_i^*\) that maximizes \(S_i\) by
\(M_i^* = \frac{1}{3}\left(Y'_i - M_{-i}\right)\)
for \(Y'_i = \frac{B}{c_i +r}\).
By Proposition \ref{prop:uniquebalance}, there exists $\CB_i^*=(E_i^*, L_i^*, M_i^*, F_i^*) \in \CO(M_i^*)$ that maximizes the expected payoff for choice of balance sheets in $\CO(M_i^*)$ for all $i\in I$.  Each $\CB_i^*$ achieves the maximal expected payoff for any choice of balance sheets in $\CO_i$ because it maximizes $S_i$.
\Aron{why is this here? if this was a proof of \cref{thm:decisionB}, then this should have been placed right after \cref{thm:decisionB} \textbf{Go: We need the previous section to complete the proof, so I prefer to place the proof here to make a style for reading. Your comment make sense if you prefer a rigid style.  If you do not like this style, please let me know.}}This proves Theorem \ref{thm:decisionB}.

\change{The}{Using the} formula of $M_i^*$ in Theorem \ref{thm:decisionB} 
\change{implies \change{}{that} if every Player in $I$ tries to maximize the payoff of the Nakamoto Game, then
it causes a \Aron{decrease compared to what?}decrease of the payoff for all the Players.   We}{we} find the equilibrium point
\(\MH_i = \frac{1}{2}\left(  Y'_i - \frac{1}{m+2}\sum_{j\in I}{ Y'_j}\right) -\frac{Z}{m+2}\)
from the fixed point of the maximizing condition of $S_i$ for each $i\in I$.  Namely, 
\(\MH_i = \frac{1}{3}\left(Y'_i - \MH_{-i} \right)\)
for $\MH_{-i} = -\MH_i + \sum_{j \in I}{\MH_j} + Z$.
This concludes Theorem \ref{thm:equilibrium}.  This equation also gives the share of the world hash rate for each Player.
See Appendix \ref{sec:sharedetail} for the proof.

\begin{corollary}\label{cor:share}
Let $\HH = \sum_{i\in I}{\MH_i } + Z$\change{,}{ be} the world hash rate at the equilibrium, and $Y'_i = \frac{B}{c_i+r}$.
Then, 
$$\HH = \frac{1}{m+2}\sum_{i \in I}{Y'_i} + \frac{2}{m+2}Z,\\
\text{and}\
\frac{\MH_i}{\HH} = \frac{1}{2}\left( \frac{Y'_i}{\HH} - 1 \right).$$
In particular, each Player's share of the world hash rate is decided only by the cost rate without explicitly depending on $m$ if the world hash rate at the equilibrium is given. 
\end{corollary}

%A capitalist Player $i\in I$ at the equilibrium will occupy the hash rate more than $50\%$ of the world hash rate if and only if
%$Y'_i > 4 \MH_{-i}.$

\void{
\begin{corollary}\label{cor:capacity}
If $c_i = c$ for all $i\in I$, then  $\MH_i/\HH <1/m$, so
$$\frac{\HH}{Y'}> \frac{m}{m+2}$$
for $Y' = Y'_i$.
\end{corollary}
}

\section{Implications in Practice}
\label{sec:implications}

\subsection{Example}\label{sec:example}
\Go{Downsize this section but keep the content}
As of February 2020, the real Bitcoin mining environment has the parameters as below.
\begin{enumerate}
    \item The world hash rate is about \change{$110,000,000$}{$1.1\times10^8$} TH/s.
    \item Bitcoin price is about \change{$9500$}{9,500} USD.
    \item Mining reward is $12.5$ BTC.
    \item The average of time intervals between block arrivals is about $10$ minutes.
    \item An example of the latest mining device is Anteminer S17+.  It costs about $2,200$ USD, including the power supply unit, and it generates about $73$ TH/s consuming $2900$W power.
    \item \Aron{most mining is outside of US, so this is not the most relevant data \textbf{Go: I think this is OK because we do not need a precise number. If not please propose modification.}} Electric generation charge is about $0.085$ USD per kWh.
    \item US 10-year Treasury Rate is $1.3\%$.  We ignore it because it is small compared with other costs and returns.
\end{enumerate}

Suppose \change{}{that} you are going to start a mining factory that mines 1 out of \change{}{every} 1000 new blocks.  
If the mining business is break-even, the mining cost per hash rate (TH) is about \Aron{express small and large numbers using exponents \textbf{Go: I feel strange to write financial numbers in scientific notation.  I will follow your guidance if you think this is normal in this area.}}$9500 \cdot 12.5 / (1.1\times 10^{8}) = 1.1 \times 10^{-3}$ USD.   We estimate the cost rate of $c_i$ is $80\%$ of the break-even point.
\void{
This estimation is in the reasonable range because the electric generation cost for the mining per hash rate (TH) is about \change{2900/1000\cdot 0.085\cdot (10/60) /73=}{$ 5.6\times 10^{-4}$} USD, so the electric cost is about $52\%$ of the break-even point.  
The depreciation of the mining device in 3 years per hash rate (TH) is about \change{2200/(3\cdot 365.25\cdot 24\cdot (60/10))/73 =}{$ 1.9\times 10^{-4}$} USD, so it is about $18\%$ of the break-even point. 
}

The world \change{Mining Assets}{mining assets} is $(1.1\times 10^{8})/73 \cdot 2200 = 3.3 \times 10^{9}$\Aron{again, this number is unreadable, use scientific number notation \textbf{Go:Updated for readability and space saving. thank you.}} USD assuming the homogeneous Players.
Your \change{Mining Assets}{mining assets} will be $(3.3 \times 10^{9})/(1-0.001)\cdot 0.001= 3.3 \times 10^{6}$ USD. It is equivalent to about $1500$ units of Anteminer S17+.
The return rate over your \change{Mining Assets}{mining assets} when mining is successful is $u =  (9500 \cdot 12.5 - (1.1 \times 10^{8}) \cdot 0.001 \cdot (1.1 \times 10^{-3}) \cdot 0.8) / (3.3 \times 10^{6}) = 3.6 \times 10^{-2} $, and for unsuccessful mining $d = ( - (1.1 \times 10^{8}) \cdot 0.001 \cdot (1.1 \times 10^{-3} \cdot 0.8) / (3.3 \times 10^{6}) = - 2.9 \times 10^{-5} $.  Applying to $f^* = \frac{u p + d(1-p)}{(u^2p + d^2(1-p)) - (up + d(1-p))^2}$, $f^*$ is about $5.6$.

\void{
A result of numerical simulation is shown in Figure \ref{fig:my_label2}.  The figure shows that choosing the optimal $f^*$ maximizes the upside gain and minimizes the downside loss.
\begin{figure}
    \centering
    \includesvg[width=8cm]{Logreturn-Simulation-Recorrected.svg}
    \caption{Profitability of the 0.1\%-success Bitcoin miner}
    \label{fig:my_label2}
\end{figure}
}
The optimal log-return rate is about $2 \times 10^{-5}$ per $10$ minutes on average.  This means the annualized return is about 180\%.
$f^*=5.6$ means you should start with about 600,000~USD for Equity, 2,700,000 USD for Liabilities. All the assets are allocated for Mining Assets, 3,300,000 USD.

\subsection{Larger is Not Necessarily Better}
Suppose that a player $i$ needs to achieve a given probability $p_i$ of successful mining and \change{optimally chooses}{chooses the optimal} $f^*$ under that constraint. 
\void{The relationship between  $f^*$ and $p_i$ under the environment of the previous subsection is  shown in Figure \ref{fig:my_label3}.
\begin{figure}
    \centering
    \includegraphics[width=6.8cm]{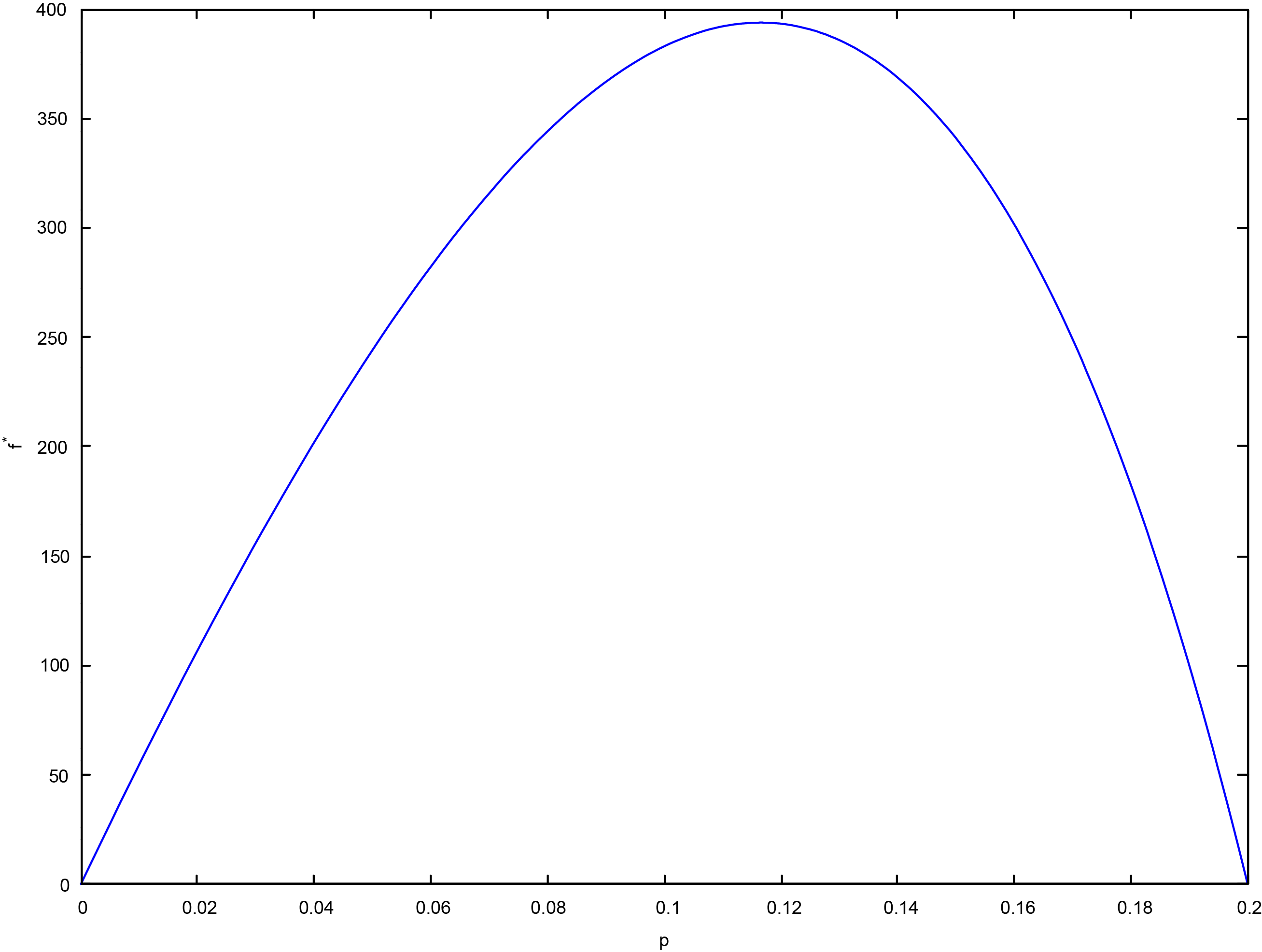}
    \caption{ The optimal $f^*$ for the success probability of mining}
    \label{fig:my_label3}
\end{figure}
\Go{ Figure 3 and Figure 4 will be removed for the space constraint.}
 We observe that }
 $f$ increases when $p_i$ is small, but decreases for larger $p_i$ and drops to $0$ at $p_i =0.2$: if we add more than $20\%$ of the world hash rate, then the Nakamoto Game becomes unprofitable because the world hash rate exceeds the break-even point.
\void{The log-return rate is shown in Figure \ref{fig:my_label4}.   

When $p_i$ is small the log-return rate increases, attains the maximum at about $p_i = 0.07$, and drops to $0$ at $p_i=0.2$.  
\begin{figure}
    \centering
    \includegraphics[width=7cm]{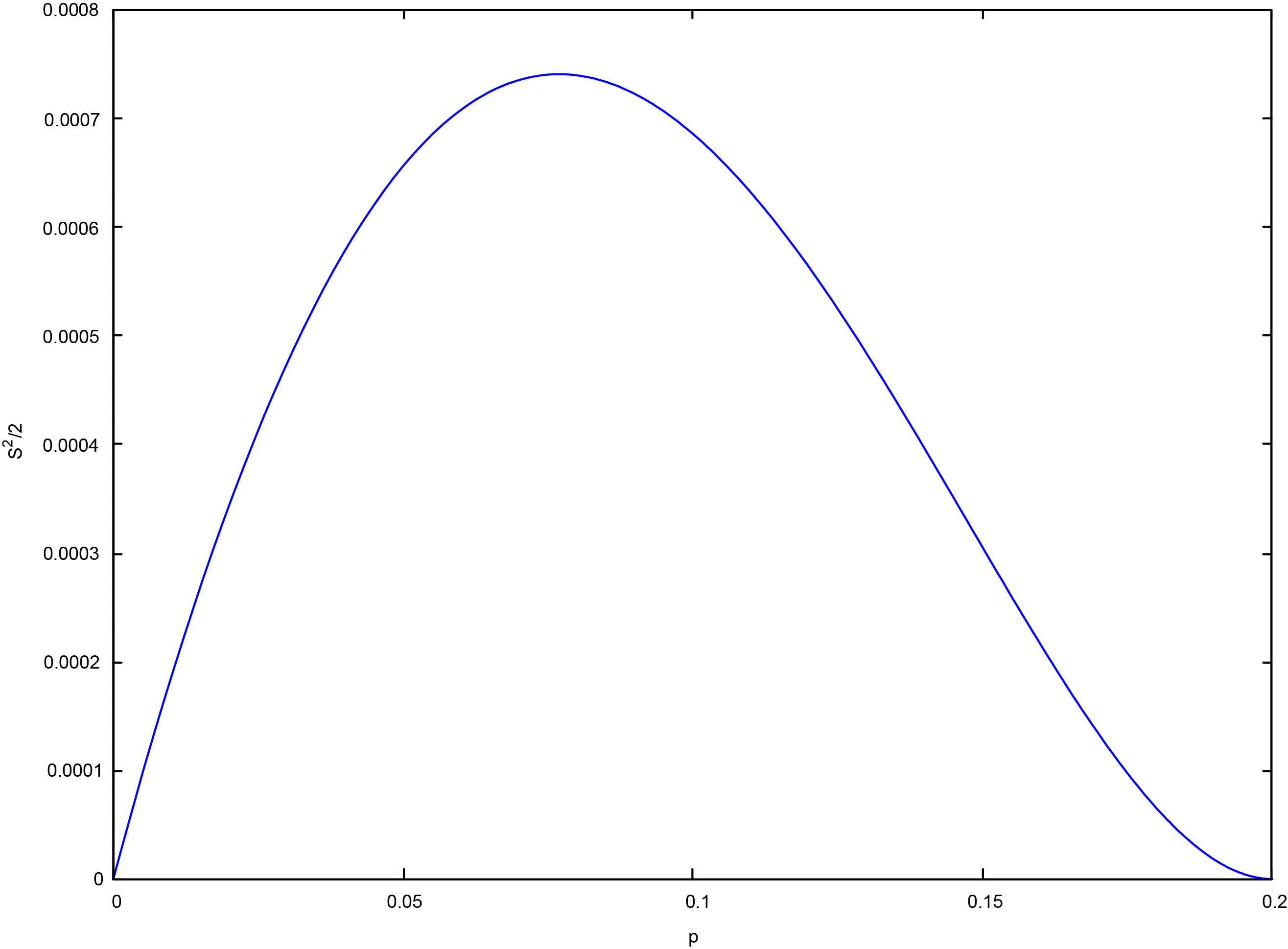}
    \caption{ The log-return rate for the success probability of mining}
    \label{fig:my_label4}
\end{figure}
}
A player will have the motivation to implement a high leverage ratio of over $100$. \void{However, it would be unrealistic to realize such a high leverage ratio using generic financial mechanisms.} The difficulty in collecting such an amount of Mining Assets \change{revisits the reason why people would like to form mining pools.}{may be one of the reasons for forming mining pools, which we discuss in the next section.}

\subsection{Mining Pools as the Players}

When $f^*$ is high, the Player has an option to work as part of a mining pool.
Since the expected simple return rates for
miners do not depend on the size of Mining Assets while the variance is smaller as
the size of Mining Assets becomes large,
the Player can reduce risk by the following methods.

\hypertarget{header-n18}{%
\textbf{\change{Type-1 Mining pools (Risk-share pools)}{Risk-Sharing Mining Pools}}\label{header-n18}}
\Go{ Changed what we call Type-1 pool. We need this modification to align with the theorem statement}
Suppose there are a set \(P\) of Players who agree that they share
the mining reward and dividend it in proportion to the amount of Mining Assets.
 Let \(W_P = \sum_{j\in P}{R_j}/M_P\) for \(M_P = \sum_{j \in P}{M_j}\), the random variable for return rate for the sum of the returns of the \change{Single Nakamoto Game}{\NRP} for Players in $P$.  Then
\(\E[W_P]=\frac{B}{H} - c_P,\) and \(
\V[W_P]= \frac{B^2 (H - M_P)}{M_P H^2},\)
for \(H = \sum_{i}{M_i}\)  and \(c_P= \sum_{j \in P}{\frac{M_j}{M_P}c_j}\). %Note that both
%the mean and the variance do not depend on the individual amount of Mining Assets
%\(M_i\). 

$W_P$ is replicated by a Player with a balance sheet $\CB\in \CO(M_P)$ for the aggregated Mining Assets \(M_P\) with cost rate \(c_P\), and each participating miners are modeled as entities which take the part of the return according to the share of the Mining Assets.  This dividend mechanism is modeled out of the Nakamoto Game.
The Player with the aggregated Mining Assets is called the \change{Type-1}{Risk-Sharing} mining pool.

\hiddenmemo{Fuhito: At some point it would be nice to make all the notation consistent: $Z$ and $H$ are both used as the world Hash rate. $Z$ is used for something else in Theorem 3. Speaking of which, notation $m$ is used as the number of (strategic) miners, but also as $\mathbb E[W_i]$ in Section 3.1.  
\textbf{Go: In a very old version $Z$ is used to be the break-even point but it is replaces with $Y$.  About $m$, I replaced it with $\mu$.}
}

In \change{reality}{practice}, \change{Type-1}{Risk-Sharing} mining pools \change{are}{were} first implemented with a proportional reward policy. 
However, it is hard to implement a fair
method that verifies each member's contribution to the mining pool's
hash rate, and many deceiving methods are proposed, such as \change{the Hopping}{hopping}
attacks. Practically, the most popular implementation \change{of the Type-1 mining
pool}{}at this moment is the pay-per-last-$N$-share (PPLNS) pools.

\void{
The return rate from mining pools with
PPLNS policy will not be exactly the same as the above formula. In this paper, we
assume we have an ideal implementation of the Type-1 mining pool for the
sake of simplicity.
}

\hypertarget{header-n42}{%
\textbf{\change{Type-2 Mining pools (Risk-free reward pools)}{Risk-Free Reward Mining Pools}}\label{header-n42}}
Suppose there is a player \(i\) who \change{prepares}{accumulates} hash rate by collecting
contributions by the external collaborators that receive \change{the}{a} risk-free fee.
\change{We call this type of Player \change{the}{a} Type-2 mining pool. }{}

Suppose Player $i$ offers mining reward of $c$ per hash rate per the average time interval for mining a new block, and \Aron{``collects hash rate performed by collaborators'' needs to be rephrased}collects hash rate $\Phi_i(c)$ performed by the collaborators.  Then the pool's success probably of the \change{single Nakamoto Game}{\NRP} is
$p_i = \frac{\Phi_i(c)}{\Phi_i(c)+M_{-i}},$
so $\Phi_i(c) = \frac{p_1}{1-p_i} M_{-i}$.
We can calculate the optimal balance sheet $\CB_i(p_i) = (E_i(p_i), L_i(p_i), M_i(p_i), F_i(p_i)) \in \CO(\frac{p_i}{1-p_i} M_{-i})$ that produces success probability $p_i$ by applying Proposition \ref{prop:uniquebalance} for $\tilde{M} = M_i(p_i) = \frac{p_i}{1-p_i} M_{-i}$.

If the mining pool wants the best log-return rate for a given $c$, then it should prepare Equity $E_i(p_i)$ determined by $p_i$, but does not need to prepare Liabilities because the Mining Assets are already levered at the optimal ratio.  The money prepared as Equity works as the reserve to pay the reward at a cost rate $c$ to collaborators.
The return from the mining pool is replicated an ordinary player with balance sheet $\CB_i(p_i)$ and cost rate $c_i = c$ but with an extra revenue of $r L_i(p_i)$, the interest of Liabilities at the riskfree rate.

\void{
If one wants to start a mining pool using a given amount of Equity $\tilde{E}$, then solve the above equations and find the appropriate $c(\tilde{E})$ such that $\CB_i(\tilde{E})$ gives $E_i = \tilde{E}$.  Then, collect hash rate $\Phi_i(c(\tilde{E}))$ by paying a reward to collaborators according to the cost rate of $c(\tilde{E})$.  It realizes the same performance as a player with balance sheet $\CB_i(\tilde{E})$ that coincides $E_i = \tilde{E}$ and cost rate $c_i =c(\tilde{E})$, but no interest payment for Liabilities.
}

In \change{reality}{practice}, this type of mining pools is implemented as the pay-per-share (PPS) mining
pools. 
It is a separate interesting topic of how we model the market of tradable hash rates that gives fair $\Phi_i$.

\subsection{Predicting Mining Pools' Shares}
\void{We have seen mining pools are also modeled as Players in the Nakamoto Game.
Theorem \ref{thm:equilibrium} predicts the share of the major mining pools.}
\void{We call mining pools that follow the optimal asset allocation strategy capitalist mining pools.}
According to practitioners, \Aron{how is this break-even price defined? for mining to be profitable for all participating miners? \textbf{Go: It is the world break-even price since the article discusses the average.}}the world break-even price of Bitcoin is about 8,000 USD as of November 2019 \cite{breakeven2019}.
\change{It}{This} implies \change{}{that} the world hash rate is estimated to be about 80\%--85\% of the break-even hash rate.
Corollary \ref{cor:share} implies each of growth-rate mining pools will have about 9\%--13\% of the hash rate \change{share}{}assuming \change{all  have the cost rate of the same range}{that they have similar \Aron{have we defined the mining pools cost rate? is it the average cost of rate of the members? \textbf{Go: Yes. It is the weighted average for Risk-sharing, and purchase price for Risk-free Reward}}cost rates}.
\void{
if the world hash rate is about $85\%$ - $80\%$ of the break-even hash rate.
If some capitalist mining pools are cost-efficient as we estimated in Section \ref{sec:example}, the current world hash rate is about $80\%$ of the break-even point, then each of them will occupy about $13\%$ of the world hash rate.  
}

\void{
It is interesting that those predictions coincide roughly with reality.
Assuming all the mining pools have a similar level of cost-efficiency, the world's capacity of the capitalist mining pools is about $11$ by Corollary \ref{cor:capacity}.  It means if you start the $12$-th capitalist mining pool, then it will increase the ratio of the world hash rate over the break-even hash rate to more than $85\%$.
}
\void{
The observation also predicts the bound of cost-efficiency that allows the $51\%$ attack by a capitalist miner. 
}
%In the present mining environment described in \ref{sec:example}, 
If a mining pool is exceptionally cost-efficient by more than about $70\%$ to the other mining pools, then \void{Corollary \ref{cor:share} indicates} the pool has a reason to occupy more than $50\%$ of the hash rate.

\void{
because $Y'_i >  (1+70\%) \frac{\HH}{85\%}= 2\HH$.
Thus the mining pool will have an economic motivation to occupy more than $50\%$ of the world hash rate, and it could be achieved in secret.
It is hard to detect such an efficient mining pool if the pool presents its hash rate divided into several mining pools with a share of about $9\%$ since it looks a set of independent mining pools with the ordinary cost rate.  
}

\section{Conclusions\change{ and Future Work}{}}
\label{sec:concl}
\Go{Downsize this section.  Maybe 10 lines}
We analyzed how the return from blockchain mining is optimized by dynamic adjustment of the asset allocation for mining resources and of financial structures for \change{the}{}mining businesses.
We have observed \change{}{that for each miner,} \change{}{how} the \change{reasonable}{optimal} share of the hash rate is \change{estimated}{determined} by the mining reward and the mining cost\change{ for each miner}{}.

\printbibliography
%\appendices
%
%\section{Experimental data}
%\lipsum[1-4]
%\section{Theorem proofs}
%\lipsum[5-6]
%
%\vspace{12pt}
\appendix
%\input{C0_rewardmodel}

%\appendix
%\input{A0_asymptotic}
\section{Proof for the Approximation in Section \ref{OptimalF}}\label{appendix:approx}

This section gives proofs of the approximation formula for $\E[\log(1+X_i)]$.
First, we can write the optimal $f$ explicitly:

\begin{align*}
\E[\log(1+ X_i)] &= p_i \log(1 + (1-f) r + f u_i) + (1-p_i)\log(1+ (1-f)r + f d_i)
\end{align*}
%\frac{M_{i}\,\log\left(-\frac{\left(M_{i}\,f-M_{i}\right)\,r+\left(M_{i}\,c-B\right)\,f-M_{i}}{M_{i}}\right)}{M_{i}+M_{-i}} +\frac{M_{-i}\,\log\left(\left(1-f\right)\,r-c\,f+1\right)}{M_{i}+M_{-i}}
for
$$
p_i = \frac{M_i}{M_i + M_{-i}},
$$
$$
u_i = \frac{B}{M_i} - c,
$$
and
$$
d_i = -c.
$$
By a straight-forward calculus we obtain
\begin{align*}
     f_\text{max} & =\frac{M_{i}\,\left(r+1\right)\,\left( (M_i + M_{-i}) (c+r)-B\right)}{\left(M_{-i}+M_{i}\right)\,\left(r+c\right)\,\left(M_{i}\,(r+c)-B\right)}.
\end{align*}

Let
$$
\mu = \E[W_i] = \frac{B}{M_i + M_{-i}} -  c
$$
and
$$
\sigma^2 = \V[W_i].
$$

%We approximate $f_\text{max}$ and $\E[\log(1+X_i)]$ when $\mu$ is small.
By eliminating $M_i$ we obtain
\begin{align*}
f_\text{max} &= \frac{ \left(B- M_{-i}\,(\mu+c) \right)\,\left(r+1\right)\,\left(\mu-r\right)}{\left(c+r\right)\,\left(M_{-i}\,(\mu + c) (c + r) +B\,(\mu - r)\right)},
\end{align*}
so we have this lemma.
\begin{lemma}
If $0 < r < \mu$, then
\begin{align*}
    f_\text{max} &< \frac{B (r+1) (\mu -r)}{M_{-i} c^3}.
\end{align*}
\end{lemma}

We are interested in approximating the optimal the log-return, so we may assume $f = O(\mu-r)$ as $\mu \to r$.

\begin{prop}
Assuming $0< r <\mu$ and $f = O(\mu-r)$ as $\mu \to r$, we have
$$
\E[\log(1+X_i)] = \log(1+r)  + f\frac{\mu -r}{1+r} - f^2 \frac{\sigma^2}{2(1+r)^2}  + O((\mu-r)^3).
$$
as $\mu \to r$.
\end{prop}
\begin{proof}
First we see
$$
\E[\log(1+X_i)] = \log(1+r) + \frac{\E[X_i -r]}{1+r} - \frac{\E[(X_i-r)^2]}{2(1+r)^2} + O((\mu-r)^3)
$$
as $\mu\to r$.
Since $\log(1+r+x) -(\log(1+r) +  \frac{(x-r)}{1+r} - \frac{(x-r)^2}{2(1+r)^2}) < \frac{(x-r)^3}{3(1+r)^3}$ for all $x> -1$, we have
$\E[\log(1+r + (X_i-r))] - (\frac{\E[X_i-r]}{1+r} - \frac{\E[(X_i-r)^2]}{2(1+r)^2}) < \frac{\E[(X_i-r)^3]}{3 (1+r)^3}.$
It suffices to show $\E[(X_i-r)^3]= O((\mu-r)^3)$.
Eliminating $M_i$, we have $p_i = 1- \frac{M_{-i}(\mu + c)}{B}$ 
%and $1-p_i = \frac{M_{-i}(\mu + c)}{B}$.
%Since $p_i = 1-\frac{ M_{-i}(c+\mu)}{B}$, 
and $u_i = \frac{B(c+\mu)}{B-M_{-i}(c+\mu)}-c$.  Since $f = O(\mu-r)$ and $0 < r< \mu$,
there exists $C=C(c, r, M_{-i})$ that satisfies
\begin{align*}
    \E[(X_i-r)^3] &= p_i ( (-f r + f u_i)^3 + (1-p_i) ( -fr + f d_i)^3\\
    &\le  C (\mu-r)^3
\end{align*}
for sufficiently small $\mu-r$.
%
%Thus we have 
%$$
%\E[\log(1+X_i)] = \E[X_i] - \frac{\E[X_i^2]}{2} + O(\mu^3).
%$$

Since $\E[X_i -r] =  f (\mu -r)$ and $\V[X_i-r] = f^2\sigma^2$, we have
%$$
%\E[\log(1+X_i)] = \log(1+r) + \frac{\E[X_i-r]}{1+r} - \frac{\V[X_i-r ]}{2(1+r)^2}  + O((\mu-r)^3),
%$$
%so
$$
\E[\log(1+X_i)] = \log(1+r)  + f\frac{\mu -r}{1+r} - f^2 \frac{\sigma^2}{2(1+r)^2}  + O((\mu-r)^3).
$$
\end{proof}

By choosing $f=f^*$ that maximizes 
$ \log(1+r)  + f\frac{\mu -r}{1+r} - f^2 \frac{\sigma^2}{2(1+r)^2}$
we obtain this corollary.

\begin{corollary}
$$
f^* = \frac{(\mu -r)(1+r)}{\sigma^2}
$$
gives an approximation $f^* = f_\text{max} + O((\mu-r)^2)$.   It satisfies
$$
\E[\log(1+X_i(f^*))] = \log(1+r) + \frac{S^2}{2} + O((\mu-r)^3)
$$
for $S = \frac{\mu -r}{\sigma}$
\end{corollary}

By taking $r \to 0$ in addition to $\mu -r \to 0$, we obtain a simpler approximation formula that works when $r$ is small.

\begin{corollary}
We have 
$$
\E[\log(1+X_i)] = r- \frac{r^2}{2}  + f{(\mu -r)} - f^2 \frac{\sigma^2}{2}  + O(\mu^3)
$$
as $r\to 0$ and $\mu-r \to 0$,
and
$$
f^* = \frac{(\mu -r)}{\sigma^2}
$$
gives an approximation $f^* = f_\text{max} + O(\mu^2)$.   It satisfies
$$
\E[\log(1+X_i(f^*))] = r- \frac{r^2}{2} + \frac{S^2}{2} + O(\mu^3).
$$
\end{corollary}

%\appendix
\section{Nash Equilibrium for Sharpe Ratio}\label{sec:sharedetail}
This section gives the calculation for the Nash equilibrium for the Sharpe ratio.

For each $i\in I$, the Shape ratio is 
\[S_i = \frac{1 - \frac{c_i + r}{ B}(M_i  + M_{-i})}{\sqrt{M_{-i}/M_{i}}}\].

Let $(\MH_1, \MH_2, \cdots, \MH_m)$ be the strategy vector at the equilibrium point.  Then it satisfies for all $i\in I$
\begin{align}\label{appendixB:equilibrium}
\MH_i &= \frac{1}{3}(Y'_i - \sum_{j \in I, i\ne j}{\MH_j} - Z)    
\end{align}

for $Y'_i = \frac{B}{c_i + r}$ because $\frac{\partial S_i}{\partial M_i} = 0$ at $M_i = \MH_i$.
By taking sum of Equation (\ref{appendixB:equilibrium}) over $i\in I$, we obtain
$$
\sum_{i\in I}{\MH_i} = \frac{1}{3}\left(\sum_{i\in I}{Y'_i}  - (m-1) \sum_{i \in I}{\MH_i} - m Z\right),
$$
so
$$
\sum_{i\in I}{\MH_i} = \frac{1}{m+2}\sum_{j\in I}{Y'_j} - \frac{m}{m+2}Z$$
hence
\begin{align*}
\HH &= Z + \sum_{i\in I}{\MH_i}\\
&= \frac{1}{m+2}\sum_{i\in I}{Y'_i} + \frac{2}{m+2}Z.
\end{align*}

Equation (\ref{appendixB:equilibrium}) also means
\begin{align*}
(1-\frac{1}{3})\MH_i &= \frac{1}{3}(Y'_i - \sum_{j \in I}{\MH_j} - Z)   \\
&= \frac{1}{3}(Y'_i - \HH),
\end{align*}
hence we have
\begin{align*}
\MH_i &= \frac{1}{2}(Y'_i -\HH)\\
&= \frac{1}{2}\left(Y'_i - \frac{1}{m+2}\sum_{j\in I}{Y'_j}\right) - \frac{1}{m+2}Z.
\end{align*}

\end{document}